\documentclass[11pt]{article} 
\usepackage[dvips]{graphicx} 
\usepackage[dvips]{color} 
\usepackage{amsmath, amsthm, amssymb} 
\usepackage{epsfig}

\newtheorem{theorem}{Theorem} 
\newtheorem{corollary}{Corollary} 
\newtheorem{lemma}{Lemma}

\newcommand{\comment}[1]{}

\title{Sleeping on the Job: \\ Energy-Efficient Broadcast for Radio Networks}
\author{Valerie King\thanks{Department of Computer Science, University of Victoria, BC, Canada; email: {val@cs.uvic.ca}} \and Cynthia Phillips\thanks{Sandia National Laboratories, NM, USA; email: caphill@sandia.gov } \and Jared Saia\thanks{Department of Computer Science, University of New Mexico, NM, USA; email: {saia}@cs.unm.edu.  This research was partially supported by NSF CAREER Award 0644058 and NSF CCR-0313160.} \and Maxwell Young\thanks{David R. Cheriton School of Computer Science, University of Waterloo, ON, Canada; email: {m22young}@cs.uwaterloo.ca}}

\begin{document}
\date{}
\maketitle

\begin{abstract}

Power is one of the most critical resources in battery-operated 
devices. In this paper, we address the problem of minimizing power 
consumption when performing reliable broadcast on a radio network. 
We consider the following popular model of a radio network. Each 
node in the network is located on a point in a two dimensional 
grid, and whenever a node sends a message, all awake nodes within 
$L_{\infty}$ distance $r$, for some fixed $r$, receive the 
message. In the \emph{broadcast problem}, some node wants to send 
a message to all other nodes in the network. We want to do this 
even when up to a $1/2$ fraction\footnote{There is nothing special 
about the fraction $1/2$. In fact, our results will still hold, up 
to constant factors, for any fixed fraction of deleted nodes that 
is strictly greater than $0$.} of the nodes within any $2r+1$ by 
$2r+1$ square in the grid can be deleted by an adversary. The set 
of deleted nodes are carefully chosen by the adversary to foil our 
algorithm and moreover, the set of deleted nodes may change 
periodically. This adversary models worst-case behavior due to 
mobile nodes moving around in the grid; static nodes loosing power 
or ceasing to function; or simply some points in the grid being 
unoccupied by nodes.

A trivial solution to this broadcast problem requires each node in 
the network to be awake roughly $1/2$ the time, and a trivial 
lower bound shows that each node must be awake for at least a 
$1/n$ fraction of the time where $n = (2r+1)^{2}$. Our first 
result is an algorithm that requires each node to be awake for 
only a $1/\sqrt{n}$ fraction of the time in expectation. Our 
algorithm achieves this reduction in power consumption even while 
ensuring correctness with probability $1$, and keeping optimal 
values for other resource costs such as latency and number of 
messages sent. We give a lower-bound that shows that this 
reduction in power consumption is asymptotically optimal when 
latency and number of messages sent must be optimal. However, if 
we can increase the latency and messages sent by only a $\log^{*} 
n$ factor we can further increase the energy savings. In 
particular, we give a Las Vegas algorithm that requires each node 
to be awake for only a $(\log^{*} n)/n$ expected fraction of the 
time in this scenario. We also give a lower-bound showing that 
this second algorithm is near optimal.

In the process of achieving these results, we define and study a 
new and compelling data streaming problem that may have 
applications in other domains. Finally, we show how our results 
can be used to ensure energy-efficient broadcast even in the 
presence of Byzantine faults.

\end{abstract}

\pagebreak

\section{Introduction}

We consider the problem of broadcast of a message over a radio 
network. We use a model for radio networks that has been studied 
extensively in the distributed computing literature~\cite{koo, 
bhandari, bhandari2, koo2}. In particular, we assume each node is 
situated on a point in a (possibly infinite) two dimensional grid 
and whenever a node sends a message, all awake nodes within 
$L_{\infty}$ distance $r$, for some fixed $r$, receive the 
message. If two nodes broadcast simultaneously, the messages 
interfere, so nodes in the neighborhoods of both senders receive 
no message. Thus, as in previous work, we assign a predetermined 
schedule of time slots to the nodes in a given neighborhood to 
avoid such message collision. We also assume that in any $2r + 1$ 
by $2r + 1$ square in the grid, up to $t$ nodes may suffer faults. 
We consider the cases where these faults are either all 
\emph{fail-stop}: the $t$ nodes are all deleted from the network; 
or, much harder, \emph{Byzantine}: the $t$ nodes are taken over by 
an adversary and deviate from our protocol in an arbitrary manner. 
The goal is to design a protocol that allows for a single node to 
broadcast a message to all other nodes in the network, so that 
eventually all non-faulty nodes learn the correct message.

Power consumption is one of the most critical resource costs in 
radio and sensor networks, particularly since nodes are usually 
battery powered. The wireless network cards on board radio network 
devices offer a number of different modes with such typical states 
as {\it off, sleeping, idle, receiving} and {\it 
sending}~\cite{armstrong:wakeup, wang:realistic} and the energy 
costs across these modes can vary significantly.
The device is always listening for messages while in the idle 
state, thus maintaining the idle state is only slightly less 
costly then maintaining the receiving state. Remarkably, the cost 
of the idle, receiving and sending states are all roughly 
equivalent, and these costs are an order of magnitude larger than 
the cost of the sleep state\footnote{The difference in energy 
consumption between the idle/send/receive states and the sleep 
state differs depending on the type of card and the communication 
standard being employed. For example, using the IEEE 802.11 
standard with a $11$ Mbps card, the ratios between power 
consumption of the idle/send/receive states and the sleep state 
are all more than $15$~\cite{feeney:investigating}. 
In~\cite{hill:system}, with a different setup employing TinyOS and 
a TR1000 transceiver, the measured ratios are over $1000$.}. Thus, 
to a first approximation, the amount of time spent in the sleep 
state gives an excellent estimate of the energy efficiency of a 
given algorithm~\cite{wang}. Unfortunately, all past algorithms 
for the reliable broadcast problem essentially ignore energy 
efficiency by never allowing any node in the network to transition 
to the sleep state\footnote{When we speak of ``reliable 
broadcast'' we are referring to the model considered in 
~\cite{koo, bhandari, bhandari2, koo2}; the terminology has been 
used in the context of other sensor network models.}.

In this paper, we directly address the problem of designing 
energy-efficient algorithms for broadcast. Crucial to our approach 
is the analysis of a new data streaming problem that we call the 
Bad Santa problem.

\subsection{Bad Santa}

A child is presented with $n$ boxes, one after another. When given 
each box, the child must immediately decide whether or not to open 
it. If the child decides not to open a box, he is never allowed to 
revisit it. At least half the boxes have presents in them, but the 
decision as to which boxes have presents is made by an omniscient 
and adversarial Santa that wants the child to open as many empty 
boxes as possible. The poor child just wants to find a single 
present, while opening the smallest expected number of boxes. This 
is the \emph{Bad Santa problem}.

More formally, a sequence of $n$ bits is streaming by and at least 
half of the $n$ bits are $1$. We can query any bit as it passes, 
but if we allow a bit to pass without querying it, it is lost. We 
must locate a $1$. That is, we seek a Las Vegas algorithm with a 
hard guarantee of success and an minimum expected number of 
queries. We are interested in two variants of this problem. First 
the single round case described above. Second the multi-round case 
where there are multiple $n$-bit streams that we query 
consecutively. In each stream at least half the values are $1$, 
but these values may be distributed differently in each stream. 
Here again we want to find how many expected queries are required 
until we find a $1$.

The connection of this problem to energy efficiency in sensor 
networks is as follows. The value $n$ is $(2r+1)^{2}$ and the bits 
to be queried represent time steps at which a message may be sent. 
During some of these time steps, due to a fault in the processor 
that is scheduled to send in that time step, no message is sent. 
In this case, the bit queried in that time step is a $0$. However, 
in at least half of the $n$ time steps, the processor is not 
faulty and in this case, the bit queried will be a $1$. The goal 
of the listening processor is to sleep as much as possible (query 
as few bits as possible) but still guarantee it will receive the 
message (query a $1$).

We insist that the data streaming algorithm guarantee that a $1$ 
be found since any probability of error would depend on $n$, which 
may be much smaller than the total number of nodes in the network. 
In particular, even a probability of error that is exponentially 
small in $n$ might not be large enough to use a union bound to 
show that all nodes in the network receive the correct message 
with high probability.

\subsection{Our Results}

We present four major results in this paper that are summarized in 
the theorems below. Theorem~\ref{t:single} is given in 
Section~\ref{single_stream}; Theorem~\ref{t:multiple} in 
Section~\ref{multi_streams}; Theorem~\ref{t:fail} in 
Section~\ref{}; and Theorem~\ref{t:byzantine} in 
Section~\ref{reliable_broadcast}.

\begin{theorem}
\label{t:single} For the single round Bad Santa problem, the 
optimal expected number of queries is $\Theta(\sqrt{n})$
\end{theorem}

\begin{theorem}
\label{t:multiple} For the $k$ round Bad Santa problem, the 
optimal expected number of queries is $O(\log^{(k)} (n/2) +k$ and 
$\Omega(\log^{(2k)} n )$. In particular, for $k = log ^{*} n$, we 
can ensure the expected number of queries is $O(log^{*} n)$
\end{theorem}

\medskip

The following two theorems about energy-efficient broadcast are 
established by algorithms based on solutions to the Bad Santa 
problem. Theorem~\ref{t:fail} essentially follows directly from 
Theorems~\ref{t:single} and~\ref{t:multiple}. 
Theorem~\ref{t:byzantine} requires a fingerprint of the message to 
first be broadcast through the network. This fingerprint must be 
sent without any nodes sleeping in order to ensure that all nodes 
receive the correct fingerprint despite the Byzantine faults. 
However, once the fingerprint is known by all nodes, the entire 
message can be sent through the network, in as energy-efficient a 
manner as in Theorem~\ref{t:fail}. For both theorems, and 
throughout the rest of this paper, we let $n = (2r+1)^{2}$.

\begin{theorem}
\label{t:fail} Assume we have a network where at most a $1/2$ 
fraction of the nodes suffer fail-stop faults in any square of 
size $2r+1$ by $2r+1$. Then there exist two algorithms, both of 
which guarantee that all non-faulty nodes in the network receive 
the correct message and which have the following properties.

\begin{itemize}
\item The first algorithm requires all nodes to be awake only a 
$1/\sqrt{n}$ fraction of the time and has optimal latency and 
bandwidth. In particular, each node sends out the message only a 
single time and each node learns the correct message within 
$O(nd)$ time steps where $d$ is the $L_{\infty}$ between the node 
and the originator of the message.

\item For any $k$ between $1$ and $\log^{*} n$, the second 
algorithm requires all nodes to be awake only a $O((\log^{(k)} n) 
/n)$ fraction of the time and has latency and bandwidth within a 
factor of $k$ of optimal. In particular, each node sends out the 
message $k$ times, and each node learns the correct message within 
$O(knd)$ time steps where $d$ is the $L_{\infty}$ distance between 
the node and the originator of the message.
\end{itemize}
\end{theorem}

\begin{theorem}
Assume we have a network where strictly less than a $1/4$ fraction 
of the nodes suffer Byzantine faults in any square of size $2r+1$ 
by $2r+1$. Further, assume that the message, $m$ to be broadcast 
consists of $|m|$ bits, and that the adversary controlling the 
Byzantine nodes would like to replace the message $m$ with some 
other message $m'$. Then there exist two algorithms, both of which 
ensure that with probability at least $1 - \frac{1}{|m|^{\log 
|m|}}$, all non-faulty nodes receive the correct message $m$. Both 
algorithms require all nodes to be awake for every step during 
which a fingerprint of size $log^{2} |m|$ is initially broadcast 
to the network. However, after this, when the message $m$ itself 
is broadcast, the algorithms have the following properties.

\begin{itemize}
\item The first algorithm requires all nodes to be awake only a 
$1/\sqrt{n}$ fraction of the time and has optimal latency and 
bandwidth.

\item For any $k$ between $1$ and $\log^{*} n$, the second 
algorithm requires all nodes to be awake only a $O((\log^{(k)} n) 
/n)$ fraction of the time and has latency and bandwidth within a 
factor of $k$ of optimal.
\end{itemize}

\label{t:byzantine}
\end{theorem}

\subsection{Related Work}

The reliable broadcast problem over a radio network arranged on a 
two-dimensional grid has been extensively studied~\cite{koo, 
bhandari, bhandari2, koo2}. The current state of the art on this 
problem is a clever algorithm that can ensure that a message is 
sent reliably to all non-faulty nodes provided that strictly less 
than a $1/4$ fraction of the nodes in any $(2r+1)$ by $(2r+1)$ 
square of the grid suffer Byzantine faults~\cite{koo2}. 
Unfortunately, as mentioned previously, all previous algorithms 
proposed for this problem require each node in the network to be 
awake for every time step, and thus are not energy-efficient. Our 
algorithm from Theorem~\ref{t:byzantine} makes use of the 
algorithm from~\cite{koo2} to broadcast the fingerprint of the 
message.

Data streaming problems have been very popular in the last several 
years~\cite{henzinger, muthu}. Generally, past work in this area 
focuses on computing statistics on the data using a small number 
of passes over the data stream. In~\cite{henzinger}, the authors 
treat their data stream as a directed multi-graph and investigate 
the space requirements of computing certain graph properties 
regarding node degree and connectedness. Munro and 
Paterson~\cite{munro} consider the problem of selection and 
sorting with a limited number of passes over one-way read-only 
memory. Along similar lines, Guha and 
McGregor~\cite{suha:approximate, suha:lower} examine the problem 
of computing statistics over data streams where the data objects 
are ordered either randomly or arbitrarily. Alon, Matias and 
Szegedy~\cite{along:frequency} examine the space complexity of 
approximating the frequency of moments with a single pass over a 
data stream. In all of these cases, and 
others~\cite{demaine:frequency, cormode:space}, the models differ 
substantially from our proposed data streaming problem. Rather 
than computing statistics or selection problems, we are concerned 
with the guaranteed discovery of a particular value, and under our 
model expected query complexity takes
priority over space complexity.\\

\section{The Single Stream Problem}\label{single_stream}

We consider the single stream problem first. A naive algorithm is 
to query $n/2 + 1$ bits uniformly at random. The expected cost for 
this algorithm is $\Theta(n)$ since the adversary will place the 
$1$'s at the end of the stream. The following is an improved 
algorithm. 

\hspace{-15pt}\begin{tabular}[t]{l p{2.8in} p{2.8in}} \hline{}&&
\end{tabular}\\
\vspace{5pt} \noindent{{\bf Single Round Strategy}}
\begin{enumerate}
\item{}Perform $\sqrt{n}$ queries uniformly at random from the first half
of the queue. Stop immediately upon finding a  $1$.
\item{}Else, starting with the first bit in the second half of the stream, query each consecutive 
bit until a $1$ is obtained.
\end{enumerate}
\begin{tabular}[t]{l p{2.8in} p{2.8in}}
\hline{}&&
\end{tabular}

\begin{theorem}
The expected cost of the above strategy is $O(\sqrt{n})$.
\end{theorem}
\begin{proof}
Assume that there are $i\sqrt{n}$ $1$s in the first half of the stream
where $i\in{}[0,\frac{\sqrt{n}}{2}]$. This implies that there are then
$(n/2) - i\sqrt{n}$ $1$s in the second half of the stream. By querying
$\sqrt{n}$ slots uniformly at random in the first half of the stream,
the probability that the algorithm fails to obtain a $1$ in the first half is no more than:
\begin{eqnarray*}
\left(1-\frac{i\sqrt{n}}{(n/2)}\right)^{\sqrt{n}} = 
\left(1-\frac{2i}{\sqrt{n}}\right)^{\sqrt{n}}
\end{eqnarray*}
\noindent{}for an expected overall cost not exceeding:
\begin{eqnarray*}
\sqrt{n} + 
\left(1-\frac{2i}{\sqrt{n}}\right)^{\sqrt{n}}\cdot{}i\sqrt{n}.
\end{eqnarray*}
\noindent{}We find the maximum by taking the derivative:
\begin{eqnarray*}
&&\frac{d}{di}\left(1-\frac{2i}{\sqrt{n}}\right)^{\sqrt{n}}\cdot{}i\sqrt{n}=\sqrt{n}\left(1-\frac{2i}{\sqrt{n}}\right)^{\sqrt{n}} 
- 2i\sqrt{n}\left(1-\frac{2i}{\sqrt{n}}\right)^{\sqrt{n}-1}
\end{eqnarray*}
\noindent{}and setting it to zero while solving for $i$ gives:
\begin{eqnarray*}
i&=&\frac{\sqrt{n}}{2(\sqrt{n}+1)}.
\end{eqnarray*}
\noindent{}Finally, plugging this into the expected cost function 
gives an expected cost of at most $\frac{3}{2}\sqrt{n}$.
\end{proof}

\noindent{}We now show that this bound is optimal within a 
constant factor.
\begin{theorem}
 $\Omega(\sqrt{n}) $ expected queries are necessary in the single round case.  
\label{thm:single-stream-lower}
\end{theorem}

\begin{proof} We follow Yao's min-max method \cite{yao:probabilistic} to prove lower bounds on any randomized algorithm which errs with probability no greater than $ \lambda=1/2^{\tilde{O}(\sqrt{n})}$:  We describe an input distribution
and show that any deterministic algorithm which errs with 
tolerance (average error) less than 
$2\lambda=1/2^{\tilde{O}(\sqrt{n})}$ on this input distribution 
requires $\Omega(\sqrt{n})$
queries on average for this distribution.  By \cite{yao:probabilistic}, this implies that the complexity of any randomized algorithm with error $\lambda$ has cost $1/2 \Omega (\sqrt{n})= \Omega (\sqrt{n})$. Let $[a,b]$ denotes the bits in position $a, a+1,...,b-1,b$ of the stream. The distribution is as follows:\\

\noindent CASE 1. With probability $1/2$, 
$\sqrt{n}$ uniformly distributed random bits  in $[1,n/2]$ are set to $1$ and the remaining bits in that interval are $0$,
$[n/2+1,n/2 + \sqrt{n}]$ are all set to
0, and the remaining bits are 1.\\

\noindent CASE 2k:  For $k=0,...,\sqrt{n}-1$, with probability $ 1/(2\sqrt{n})$,  $[1,...,n/2]$ contains 
a uniformly distributed random set of $k$ 0's and the rest are 1's. Then $\sqrt{n}-k$ 0's are contained in  uniformly distributed random bit positions in
$[n/2+1, n/2+\sqrt{n}]$, and the remaining $k$ bits in positions $[n/2+1, n/2+\sqrt{n}]$ are 1's. The remaining bits in the stream are 0.\\

\smallskip

\noindent {\it Analysis:} Let $A$ be a deterministic algorithm which errs with average probability less than $2\lambda$. Note that $A$ is completely specified by the list of indices of bits to query while it has not yet discovered a $1$, since it stops as soon as it sees a $1$.  Let $x$ be the number of queries in the list  which lie in  $[1,n/2]$.
For a constant fraction of inputs  in CASE 1,  $A$ will not find a 1 in $[1,n/2]$ 
within $\sqrt{n}$ queries. Hence either $x \geq \sqrt{n}$ or  $A$ must find a 1 with high probability 
in $[n/2,n]$.
 Now suppose $x <  \sqrt{n}$. We show that $A$'s list $L$ must contain greater than  $\sqrt{n} -x $ bit positions in $[n/2+1, n/2+ \sqrt{n}]$. If not, it will  err on the input in CASE 2x in which all the
 $x$  positions queried in $[1,n/2]$ and the $\sqrt{n} -x$ positions queried in  $[n/2, n/2 + \sqrt{n}]$ are $0$.
 Since this input occurs with probability  $(2\sqrt{n})^{-1} {n/2 \choose x}^{-1} {\sqrt{n} \choose  x}^{-1} = 2\lambda$ in the distribution, the algorithm errs with probability at least $2\lambda$ and there is a contradiction.
 We conclude that any algorithm erring with probability less than $2\lambda$ must  either  have $x> \sqrt{n}$ or queries greater than $\sqrt{n}-x$ bits of 
 $[n/2+1,  n/2+ \sqrt{n}]$. 
 
 Now, we show that any such deterministic algorithm incurs an average cost
 of $\Omega(\sqrt{n})$ on the CASE 1 strings in this distribution. If $x \geq  \sqrt{n}$ then for a constant fraction of strings in 
 CASE 1, the algorithm will ask at least $\sqrt{n}$ queries in $[1,n/2]$ without finding a 1. 
 If $x < \sqrt{n}$, then with constant probability the algorithm will incur a cost of $x$  in $[1,..n/2]$  and
 go on to incur a cost of $\sqrt{n}-x$ in $[n/2 +1,n/2+\sqrt{n}]$ since all the values there are 0.
 
We have shown that the distributional complexity with error 
$2\lambda$ is $\Omega(\sqrt{n})$. It follows from 
\cite{yao:probabilistic} that the randomized complexity with error 
$\lambda$ is $\Omega(\sqrt{n})$.
\end{proof}

\subsection{The Multiple Streams Problem}\label{multi_streams}
We define a $(\alpha, \beta)$-strategy to be an algorithm which 
occurs over no more than $\alpha$ streams, each with at least a 
(possibly different) set of at least $n/2$ values of 1, and which 
incurs expected cost (number of queries) at most $\beta$. In the 
previous section, we demonstrated a $(1, O(\sqrt{n}))$-strategy. 
We now consider the following protocol over $(k+1)$ streams.\\

\clearpage{} 
\noindent{}\begin{tabular}[t]{l p{2.8in} p{2.8in}}
\hline{}&&\end{tabular}\\
\noindent{{\bf Multi-Round Selection Strategy}}\\\\
\noindent{}For $i=k$ to $1$
\begin{itemize}
\item{}Perform $\lg^{(i)} (n/2)$ queries  uniformly at random over the entire stream. Stop if a 
$1$ is obtained.
\end{itemize}
\noindent{}If no value of $1$ has been found, use the single stream strategy on the final stream.\\
\begin{tabular}[t]{l p{2.8in} p{2.8in}}
\hline{}&&
\end{tabular}
\begin{theorem}\label{multistream}
The above protocol is a $(k+1, O(\lg^{(k)} (n/2) + k))$-strategy.
\end{theorem}
\begin{proof}
Correctness is clear because in the worst case, 
we use the correct one-round strategy in the final round. 
The expected cost is at most
\begin{eqnarray*}
&&\lg^{(k)} (n/2) + \left[\sum_{i=k-1}^{1}\left(\frac{1}{2}\right)^{\lg^{(i+1)} (n/2)}\lg^{(i)} (n/2)\right] +\left(\frac{1}{2}\right)^{\lg{(n/2)}}O(\sqrt{n})\\
&=&\lg^{(k)} (n/2) + \left(\frac{1}{2}\right)^{\lg^{(k)} 
(n/2)}\lg^{(k-1)}(n/2)
+ ... + \left(\frac{1}{2}\right)^{\lg{(n/2)}}O(\sqrt{n})\\
&=&\lg^{(k)} (n/2) + k + o(1)
\end{eqnarray*}
\end{proof}

\begin{corollary}
If there are $\lg^*{(\frac{n}{2})}+1$ streams, then the 
multi-stream algorithm provides a \\$(O(\lg^*{\frac{n}{2}}), 
O(\lg^*{\frac{n}{2}}))$-strategy.
\end{corollary}
\begin{proof}
By the definition of the iterated logarithm: 
\begin{eqnarray*} 
\lg^*{n} = \Biggl\{\begin{array}{cc}
  0 & \mbox{~~for~}n\leq{}1 \\
  1+\lg^*(\lg{n}) & \mbox{~~for~}n>1\\
\end{array}
\end{eqnarray*}
\noindent{}if $k=\lg^*{(n/2)}$, we can plug this value into the 
last line of the proof for Theorem~\ref{multistream} which 
contains three terms of which only the first two depend on $k$. 
The first term is $1$, by definition of $\lg^*{n}$, and the second 
is $k$ for a total expected cost of $1 + \lg^*{(n/2)} + o(1)$.
\end{proof}

\subsection{Lower bound for multiple streams}\label{multi_streams_lower}
First, we show the following lemma:
\begin{lemma} \label{l:firstpass}
$\Omega (\lg^{(i+2)}n)$ expected queries are required for a 
randomized algorithm that errs with probability less than $\lambda 
=(\ln^{(i)} n)^{-\epsilon}$ on one stream of length $n$. In 
particular, when $i=0$, $\Omega (\log \log n)$ expected queries 
are required for a randomized algorithm with error less than $1/ 
n^{\epsilon}$, for any constant $\epsilon$.
\end{lemma}

\begin{proof}
We apply Yao's min-max method \cite{yao:probabilistic} and 
consider the distribution in which with probability $1/3$, one of 
the $I_1=[1,n/3]$, $I_2=[n/3+1, 2n/3]$, and $I_3=[2n/3+1,n]$ 
intervals is all 0's, and the other two each contain exactly $n/4$ 
$1$'s with the $1$'s distributed uniformly at random. Let $L$ 
denote the list of of queries of a deterministic algorithm, and 
let $x_i$ be the number of queries in $L \cap I_i$. The 
probability that the algorithm fails to find a 1 in any interval 
$I_i$ is ${n/3 -x_i \choose n/4 }/{ n/3 \choose n/4} > 
e^{-7x_i/4}$. Let $I_i$ and $I_j$ be the intervals which are not 
all 0's. Then the probability of failing to find a 1 in either 
$I_i$ and $I_j$ is $> e^{-7(x_i+x_j)/4}$. Hence the probability of 
not finding a $1$ over all strings is $>(1/3) e^{-7(x_i+x_j)/4} > 
2 \lambda$ if $x_1+x_j < (3/7) \epsilon \lg^{(i+1) n}$. We 
conclude that a deterministic algorithm with average error less 
than $2\lambda$ can have at most one $x_i, i=1,2,3$ such that $x_i 
< (3/14) \epsilon \lg^{(i+1)} n$.
 
Now we examine the cost of such an algorithm. Suppose $x_1 \geq 
(3\epsilon/14)(\ln^{(i+2)}n$ then with probability $1/3$ $I_1$ is 
all $0$'s and the cost incurred is $x_1$, for an average cost of 
$(\epsilon/14)(\ln^{(i+2)}n$. Now suppose $x_1 < (3\epsilon/14) 
\ln^{(i+2)}n$. From above, we know $x_2 > (3\epsilon/14)\ln^{(i+1) 
}n$. Then with probability $1/3$, $I_2$ is all $0$'s and with 
probability $> e^{-7x_1/4} > ({\ln^{(i+1)} n})^{-3\epsilon/8}$, 
the algorithm does not find a $1$ in $I_1$ and incurs a cost of 
$(3\epsilon/14)\lg^{(i+1)} n$ in $I_2$ for an average cost of at 
least $(\epsilon/14)({\ln^{(i+1)}n})^{1-3\epsilon/8}$. Hence the 
average cost of any such deterministic algorithm is at least 
$\min\{ (1/14)(\ln^{(i+2)}n, (\epsilon/14) 
(\lg^{(i+1)}n)^{1-3\epsilon/8} \} = \Omega(\ln^{(i+2)}n)$. By 
Yao's min-max method \cite{Yao}, any randomized algorithm with 
error $\lambda$ is bounded below by $1/2$ the average cost of a 
deterministic algorithm with average error $2\lambda$ on any 
distribution. The lemma now follows.
 \end{proof}

\begin{theorem}
For $k>0$, $\Omega(\ln^{(2k)} n )$ expected queries are 
necessary to find a $1$ from  $k+1$ streams with probability $1$.
\end{theorem}

\begin{proof} The proof is by induction on the number of streams. 

\noindent{\it Base Case:} Let $k=1$. Either the algorithm finds a 
1 in the first pass or the second pass. From Lemma 
\ref{l:firstpassh}, for any constant $\epsilon$ any algorithm 
which fails to find a $1$ in the first pass with probability $\leq 
n^{-\epsilon}$ has expected cost $\Omega(\log \log n)$. If the 
algorithm fails to find a $1$ in the first pass with probability 
at least $ n^{-\epsilon} $ then the expected cost to the algorithm 
is at least the probability it fails in the first pass times the 
expected cost of always finding a $1$ in the second and final 
pass, which is $ n^{-\epsilon} * \Omega(\sqrt{n})$. (The second 
factor is from Lemma \ref{thm: single-stream-lower}. Choosing 
$\epsilon < 1/2$, the expected cost is $\Omega(\log \log n )$.

\noindent{\it Inductive Step:} Now assume the hypothesis is true for up to $k>1$ streams. Assume we 
have $k+1$ streams. Any randomized algorithm either fails to find a 1 in the first stream with probability
less than $(1/\ln^{(2k-2)})^{\epsilon}n$, in which case by Lemma \ref{l:firstpass}, the expected cost of the algorithm when it processes the first stream is  $\Omega (\ln^{(2k)}n)$ or the probability that it fails in the first pass is at least $(1/\ln^{(2k-2)})^{\epsilon}n$.  In that case,  the expected cost deriving from queries of the second stream is at least $(1/\ln^{(2k-2)})^{\epsilon}n  * \Omega(\ln^{(2k-2)} n )$ where the second factor of this expression is the expected number of queries needed to find a $1$ in $k$ streams, as given by the induction hypothesis. The minimum expected cost of any randomized algorithm is the minimum of these two possibilities, which is 
 $\Omega (\ln^{(2k)}n)$.
\end{proof}

\section{Application of the Streaming Problem to Reliable Broadcast in Radio Networks}\label{reliable_broadcast}

We now use the streaming algorithms in the radio network model
described by~\cite{koo,{bhandari},{bhandari2},{koo2}}. Under this 
model, nodes are situated in a lattice and can communicate with 
other nodes within a radius $r$, under a defined metric, via 
wireless multicast. Considering the $L_{\infty}$ metric, a 
Byzantine adversary can place faulty nodes at any location in the 
lattice subject to the constraint that every 
$(2r+1)\times{(2r+1)}$ region of nodes contains at most $t$ faulty 
nodes. Previous work has focused on securely transmitting a value 
$m$ from a correct node $s$, known as the dealer, to all nodes in 
the lattice. The algorithms of~\cite{bhandari, bhandari2} are 
optimal in the sense that they achieve reliable broadcast for 
$t<(r/2)(2r+1)$ while Koo~\cite{koo} has shown that no algorithm 
can tolerate $t\geq{}(r/2)(2r+1)$. Recent work by~\cite{koo2} 
extends these results to the setting where faulty nodes can spoof 
the addresses of correct nodes in the network and cause a bounded 
number of message collisions. Under the protocol described 
in~\cite{bhandari, bhandari2}, for a value $m$ of $|m|$ bits, the 
adversary can force a node to listen to 
$(2t+1)|m|\geq{}r(2r+1)|m|={\Theta(r^2|m|)}$ bits before 
committing.

Here we maintain the original model for finite graphs with the 
added assumption that {\it the adversary is computationally 
bounded}. Under this scenario, it is possible to achieve a 
substantial power savings over the original protocols 
of~\cite{bhandari, bhandari2}, as we will now demonstrate. Let $h$ 
be a secure hash function known to all nodes that takes an input 
$m$ and outputs a fingerprint of size $\Theta{(\log^2{|m|})}$ 
bits.

For the following analysis, we define a corridor of width $2r+1$
starting at the dealer located at point $(0,0)$ and ending at node $p
= (x,y)$.  We sometimes denote a point $q$ at grid location $(x,y)$ as
$q(x,y)$.  We define the set of nodes in the corridor to be
$S_{cor}=S_{x,cor}\cup{}S_{y,cor}$ where $S_{x,cor}=\{q(x',y') ~|~
(-r\leq{}x'\leq{}x) \wedge{} (y-r\leq{}y'\leq{}y+r)\}$ and
$S_{y,cor}=\{q(x',y') ~|~(-r\leq{}x'\leq{}r)
\wedge{}(0\leq{}y'\leq{}y+r)\}$. Figure~\ref{corridor} illustrates a
corridor for $r=3$. We use the schedule suggested in~\cite{koo}. Under
this schedule, every node in $N(x,y)$, the neighborhood of point
$(x,y)$, is allotted one time step to broadcast before the schedule
repeats after $(2r+1)^2$ time steps. We will refer to a full pass
through the $(2r+1)^2$ time steps of the schedule as a {\it round}. We
need the following lemma which is an extension of implicit claims
in~\cite{bhandari} and the statement of~\cite{koo2} (Claim 1).

\begin{figure}[t]
\begin{center} 
\hspace{0.5cm}\includegraphics[scale=0.4]{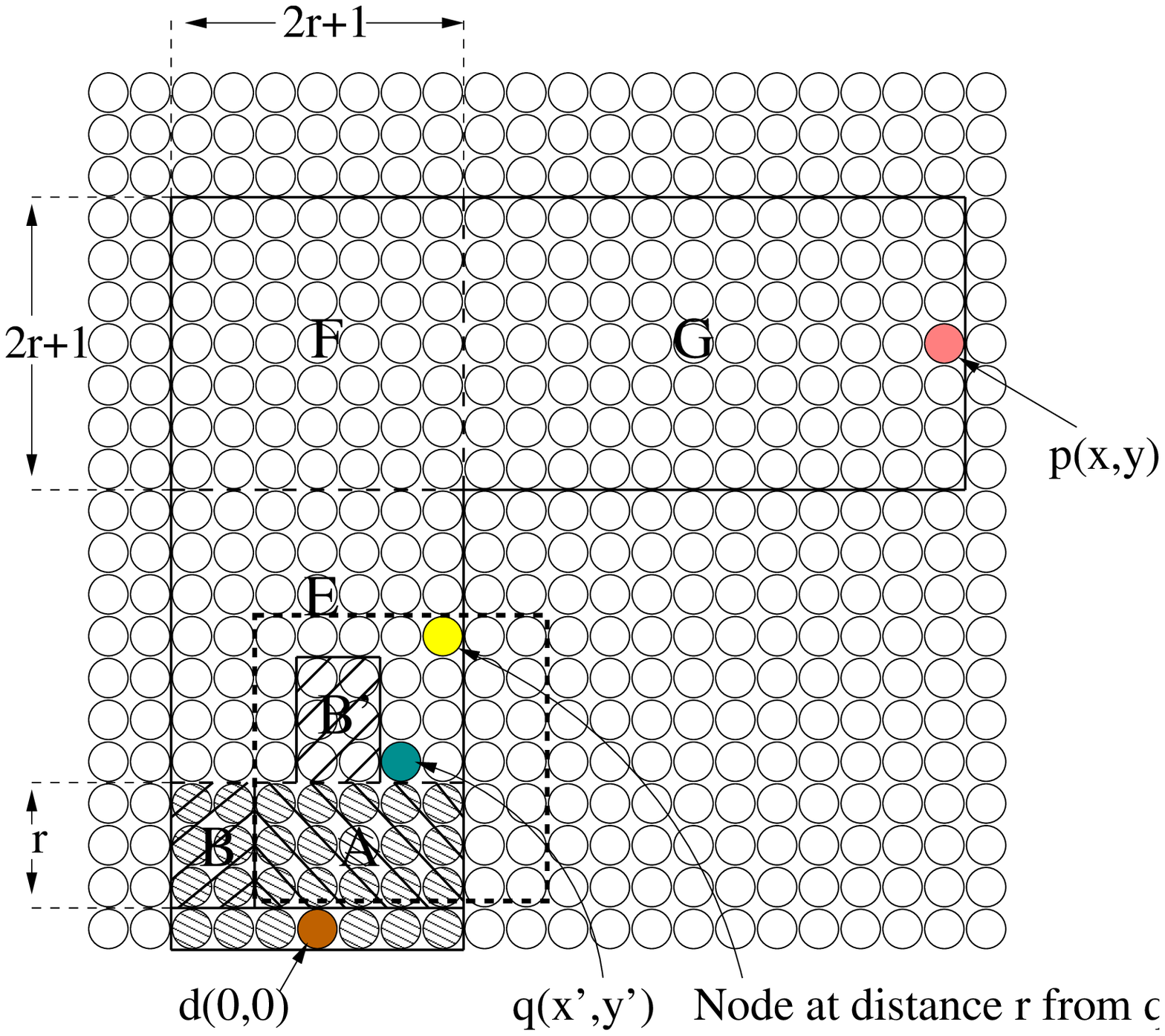}
\end{center}
\caption{A depiction of a corridor for $r=3$. Together the nodes 
in regions $E$ constitute $S_{y,cor}$ while the nodes in $F$ and 
$G$ constitute $S_{x, cor}$. Node disjoint paths of the form 
$(u_i, q)$ originate from nodes $u_i$ in region $A$. Node disjoint 
paths of the form $(u_i,u_i',q)$ originate from nodes $u_i$ in 
region $B$ and traverse through nodes in $u_i'$ in $B'$ to reach 
node $q$. The yellow node lies at a distance $r$ from $q$ and has 
a location farthest from $(0,0)$ that lies in 
$N(q)\cap{}S_{cor}$.}\label{corridor}
\end{figure}

\begin{lemma}\label{delay}
If the dealer, $d(0,0)$, broadcasts $m$ at time step $t_{init}$, 
node $p(x,y)$ is able to commit to $m$ at least by time step 
$t_{init}+2(2r+1)^2(|x|+|y|-r)$.
\end{lemma}
\begin{proof}
We are assuming the same schedule as in~\cite{koo} and the same 
protocol as in~\cite{bhandari2}; however, we are restricting our 
view to those nodes in $S_{cor}$. That is, nodes in $S_{cor}$ will 
only accept messages from other nodes in $S_{cor}$ and they will 
ignore all messages they receive from nodes outside the corridor. 
Clearly, this can only result in a slowdown in the propagation of 
the broadcast value; moreover, the rectilinear shape of the 
corridor can only slow down the rate of propagation in comparison 
to the original propagation described in~\cite{bhandari2}. An 
argument identical to that in ~\cite{bhandari2} (Theorem $4$) can 
be used to show each node $q(x',y')\in{}S_{cor}$ will commit to 
the correct value by receiving messages along at least $2t+1$ node 
disjoint paths $P_i$ of the form $(u_i, q)$ and $(u_i,u_i',q)$ 
where $u_i,u_i'$ are distinct nodes and lie in the corridor and 
$u_i\in{}N(x',y')$; we do not repeat this argument.

We now consider the time required until $p(x,y)$ can commit to 
$m$. Without loss of generality, assume that $x,y\geq{}0$ and that 
the broadcast first moves nodes in $S_{y,cor}$ and then along 
nodes in $S_{x,cor}$. At $t_{init}$, the dealer broadcasts $m$ and 
all nodes in $N(0,0)$ commit to $m$. Consider a node $q(a,r+1)$ 
where $-r\leq{}a\leq{}r$. It takes at most one round for $q$ to 
receive messages along paths of the form $(u_i, q)$. Concurrently, 
in this one round, nodes $u_i$ can transmit messages to nodes 
$u_i'$ along paths of the form $(u_i, u_i', q)$. At most an 
additional round is required to send from nodes $u_i$ to $q$. 
Therefore, at most two rounds are required before $q$ can commit. 
Note that this holds for all nodes with coordinates $(a,r+1)$ for 
$-r\leq{}a\leq{}r$; this entire row can commit after at most two 
rounds. It follows that all nodes in $S_{y,cor}$ are committed to 
$m$ after $2y$ rounds. An identical argument can be used to show 
that all nodes in $S_{x,cor}$ are committed to $m$ after $2(x-r)$ 
rounds. Therefore, $p$ commits after at most $2(x+y-r)$ rounds or, 
equivalently, $2(2r+1)^2(x+y-r)$ time steps; if $x$ and $y$ can 
take on negative values, this becomes $2(2r+1)^2(|x|+|y|-r)$.
\end{proof}

\noindent{}Our protocol for reliable broadcast is as follows:

\noindent{}\begin{tabular}[t]{l p{2.8in} p{2.8in}} \hline{}&&
\end{tabular}

\vspace{5pt} \noindent{{\bf $(k+1, O(\lg^{(k)}{(n/2)} + k))$ 
Reliable Broadcast with Bit Reduction}}
\begin{enumerate}
\item{}Initially, the dealer $d(0,0)$ does a local broadcast of $(h(m), t_{init})$. 
\item{}Each node $i$ in $N(d)$ commits to $h(m)$ and does a one-time local broadcast
of $\textsc{commit}(i,h(m), t_{init})$. \\\\ The following 
protocol is followed by each node $p$ including the nodes in the 
first two steps:
\item{}If node $p$ receives $\textsc{commit}(q,h(m), t_{init})$ for the first time, it records this
message and broadcasts $\textsc{heard}(p,q,h(m), t_{init})$.
\item{}If node $p$ receives $\textsc{heard}(q,w,h(m), t_{init})$ for the first time, it records this
message.
\item{}If at any point, $p(x,y)$ holds $t+1$ 
messages $m_1, m_2, ...$ such that $\textsc{commit}(a_i,h(m), 
t_{init})$ or $m_i=\textsc{heard}(a_i, a_{i'}, h(m), t_{init})$ 
where for all $i$, $a_i,a_{i'}\in{N(q)}$ for some node $q$ and for 
all $i,j$, $a_i\not={}a_j$, $a_i\not={}a_{j'}$, $p(x,y)$ sets 
$f_{maj}=h(m)$, does a one-time broadcast of 
$\textsc{commit}(p,f_{maj}, t_{init})$.

\item{}For $i=0, ..., k$ 
rounds, $p$ executes the following:
\begin{enumerate}
\item{}Let $G_p$ denote the set of all nodes that sent $p$ 
the majority fingerprint $f_{maj}$ and the majority initial time 
value $t_{init}$. Assume $|G_p|=n$ and, for ease of analysis, let 
us consider $G_p$ to be an ordered set where the nodes, denoted by 
$\{g_0,
..., g_{n-1}\}$, are ordered from earliest to latest by broadcast 
order as dictated by the broadcast schedule. Node $p$ picks slots 
to listen to according to the algorithm for the $(k+1, 
O(\lg^{(k)}{(n/2)} + k))$-strategy. Let $Q_i$ denote the set of 
nodes corresponding to these chosen time slots.

\item{}Node $p$ listens during the time 
slot only when 1) node $q(x', y')\in{Q_i}$ is scheduled to 
broadcast and 2) immediately following time step 
$t_{init}+2(2r+1)^2(|x'|+|y'|+r)+1$ (inclusive). In listening to 
node $q$, $p$ will obtain a value $m_q$. If $h(m_q)=f_{maj}$, then 
$p$ commits to $m_q$, breaks the for-loop and proceeds directly to 
Step 7.
\end{enumerate}

\item{}Node $p(x,y)$ waits until the time step 1) when node $p$ is scheduled 
to broadcast and 2) immediately following time step 
$t_{init}+2(2r+1)^2(|x|+|y|+r)+1$ (inclusive) and then does a 
broadcast of $\textsc{commit}(p,m_q)$. Node $p$ broadcasts 
$\textsc{commit}(p,m_q)$ for a second and final time on its next 
turn (in the next round).
\end{enumerate}
\begin{tabular}[t]{l p{2.8in} p{2.8in}}
\hline{}&&
\end{tabular}

\subsection{The Necessity of Hard Guarantees}

\noindent{}We now explicitly address the necessity of a hard 
guarantee in our reliable broadcast protocol. At first glance, it 
may appear that the hard requirement can be ignored. However, if a 
correct node fails to commit to a correct message,~\cite{bhandari} 
proves that reliable broadcast is impossible\footnote{We need to 
guarantee the existence of at least $r(2r+1)$ paths of which 
strictly less than half traverse a faulty node. If any good node 
fails, this destroys the necessary invariant that we always have a 
(narrow) majority of uncorrupted disjoint paths.}. Therefore, for 
reliable broadcast to be guaranteed with probability $1$, we 
cannot tolerate {\it any} failures. In turn, to guarantee reliable 
broadcast with high probability, the probability of such a failure 
must be overwhelmingly small in the total size of the network, 
$N$. There are two ways in which a correct node could fail to 
obtain the correct message by the above protocol:
\begin{enumerate}
\item{}The adversary achieves a collision for the secure hash function.
\item{}A good node does not query enough nodes in its neighborhood.
\end{enumerate}
\noindent{}Let $n=(2r+1)^2$ be the number of nodes in a single 
neighborhood. The first case is avoided by using a large enough 
fingerprint. In the second case, we can use a sampling algorithm 
to reduce the probability that a node fails to obtain the correct 
message. However, without a hard guarantee, the probability of 
such a failure is bounded in terms of $n$ (the size of the 
neighborhood), not $N$. Realistically, we should expect $r$ to be 
small in comparison to the total size of the network and, 
consequently, $n\ll{}N$. Therefore, in this case, we cannot make 
the probability of such a failure small in terms of the total size 
of the network and so we require a hard guarantee on message 
receipt.

\subsection{Correctness of the Bit-Reduction Algorithm}

\noindent{}Now that we have illustrated the key design points of 
our algorithm, we prove its correctness.

\begin{theorem}\label{thm_correctness}
Let $m$ denote the message sent from the dealer and let $|m|$ be 
the number of bits in $m$. Let $h$ be a secure hash function that 
maps to a fingerprint of size $\Theta(\log^2 |m|)$. The reliable broadcast protocol 
with bit reduction has the following properties:
\begin{itemize}
\item{}Reliable broadcast of $m$ is achieved with probability of error that is 
superpolynomially small in $|m|$ (i.e. less than $O(1/|m|^{C\log 
|m|})$ for some constant $C$).
\item{}If $k=1$, node $p$ listens to $O(r^2\log^2{|m|} + |m|r)$ bits in expectation. 
\item{}If $k\geq{}2$, node $p$ listens to $O(r^2\log^2{|m|} + |m|(\lg^{(k)}{r} + k))$ bits in expectation.
\end{itemize}
\end{theorem}
\begin{proof}
We begin by proving correctness. By assumption, $h$ is a secure 
hash function and, therefore, given $h(x)$, the probability that 
the adversary obtains a value $x'$ such that $h(x')=h(x)$ is 
$2^{-\Theta(\log^2{|m|})}=|m|^{-\Theta(\log{|m|})}$. Therefore, it 
will take the adversary superpolynomial time in $m$ to forge such 
an $x'$ and so $f_{maj}$ will correspond to the correct value $m$. 
Steps 1-5 of the BR protocol, are no different than the broadcast 
presented in~\cite{bhandari2} where the values being transmitted 
are a fingerprint and an initial time value. Consequently, every 
correct node will be able to derive a majority fingerprint 
$f_{maj}$ and $t_{init}$ value from the messages it receives. By 
the security of $h$, $h(m)=f_{maj}$ with high probability and, by 
the correctness of the protocol in~\cite{bhandari2}, the majority 
time value will be the true $t_{init}$ sent
originally by the dealer.\\

\noindent{}In Step $6a$, if $p$ receives $m_q=h(f_{maj})$ from $q$ 
for some $i=0, ..., k$, $p$ can be assured, again with high 
probability, that $m_q$ is the correct value $m$. Let $S_p$ denote 
the set of nodes from which $p$ can receive \textsc{commit} 
messages. A complication arises in attempting to guarantee that 
$p$ does not miss any of the $k+1$ broadcasts of $m_q$ by 
$q\in{}S_p$. For instance, an extreme case occurs if every node in 
$S_p$ broadcasts \textsc{commit} messages before $p$ begins 
sampling for the true value. In such a case, $p$ will never 
receive the actual value to check against its fingerprint. To 
ensure that the appropriate nodes do not broadcast $m_q$ before 
$p$ is ready, it is enough to impose a sufficiently long delay 
before any $q\in{S_p}$ broadcasts $m_q$; enough time for $p$ to 
first commit to $f_{maj}$. In particular, if all correct nodes in 
$N(q)$ have obtained their majority fingerprint value, then it is 
safe for $q$ to begin executing its $k+1$ broadcasts of 
$\textsc{commit}(p,m_q)$. Assuming the dealer broadcasts the 
initial fingerprint at time $t_{init}$, then by Lemma~\ref{delay} 
it will take at most $t_{init}+2(2r+1)^2(|x|+|y|-r)$ time steps 
for $p(x,y)$ to commit to $f_{maj}$. Therefore if $q$ waits until 
after this time step, it can be assured that $p$ will be able to 
receive its transmission of $m$. In general, $q$ must wait until 
all peers in $N(q)$ can commit to a fingerprint. Since all nodes 
in $N(q)$ are located at most $r$ distance away, $q(x',y')$ can 
wait until time 
$t_{init}+2(2r+1)^2((|x'|+r)+(|y'|+r)-r)=t_{init}+2(2r+1)^2(|x'|+|y'|+r) 
$ at which point all of $N(q)$ is guaranteed to have committed to 
a fingerprint. Define $T$ be the time slot that is 1) scheduled 
for node $q$ to broadcast and 2) occurs first after time step 
$t_{init}+2(2r+1)^2(|x'|+|y'|+r)+1$. By Lemma~\ref{delay} 
guarantees that $p\in{N(q)}$ will have already committed to a
majority fingerprint value and be ready to listen to $m$.\\

\noindent{}By the hard guarantee of our streaming problem 
algorithm, $p$ will obtain a message that, when hashed, matches 
the fingerprint to which $p$ committed. Therefore, $p$ is 
guaranteed to commit to a value $m_p$ and, with high probability 
in $|m|$, $m_p$ is the correct value sent by the 
dealer\footnote{We assume the dealer is correct; otherwise, $p$ is 
committing to the value agreed upon by all correct nodes in $N(d)$ 
which is also correct.}. Finally, again guaranteed by 
Lemma~\ref{delay}, $p$ only begins the first of two consecutive 
broadcasts of $\textsc{commit}(p, m_p)$ after all nodes in $N(p)$
have been given enough time to commit to a fingerprint.\\

\noindent{}To show resource costs, note that $p$ listens to 
$(2t+1)\leq{r(2r+1)}$ fingerprints, for a total of 
$O(r(2r+1)\log^2{|m|})$ bits before sampling. Now note that the 
sampling strategy described in Step $6$ matches our algorithm for 
the streaming problem\footnote{Note that selecting a random node 
is necessary; if not, the adversary might have faulty nodes send 
correct fingerprints in the first round and, if $p$ selects nodes 
from $\mathcal{B}$ in a deterministic fashion, the adversary may 
force $p$ to listen to many messages that do not hash to 
$f_{maj}$.}. While in the streaming problem, we attempt to obtain 
a $1$ at unit cost per query, here node $p$ is attempting to 
select a correct node at the cost of listening to $|m|$ bits per 
selection. Theorem~\ref{multistream} guarantees that a correct 
message is obtained at an expected cost of $O(\lg^{(k)}{(n/2)} + 
k)$ bits where $n=(2r+1)^2$. Therefore, the total expected cost is 
$O(r^2\log^2{|m|} + |m|(\lg^{(k)}{(r)} + k))$.
\end{proof}

\noindent{}Note that $|m|$ need not be very large in order to make 
the probability of a collision negligible. For instance, if 
$|m|=1$ kilobyte, the probability of a collision is already less 
than $10^{-30}$. Finally, it follows immediately from the above 
theorem that for a message $m$ of size 
$|m|\in{}\omega{}(\log^2{N})$, the bit-reduction algorithm 
achieves a substantially lower communication complexity.

\section{Future Work and Conclusion}
  
We have designed new algorithms for reliable broadcast in radio 
networks that achieve near optimal energy savings. These 
algorithms consume significantly less power than any other 
algorithms for this problem of which we are aware. In the process 
of designing and analyzing these algorithms, we have defined and 
studied a novel data streaming problem, which we call the Bad 
Santa problem.

Several open problems remain including: Can we close the gap 
between the upper and lower-bound for the multi-round Bad Santa 
problem? Can we be robust to Byzantine faults caused by a 
computationally \emph{unbounded} adversary and still reduce power 
consumption? Can we generalize our techniques to radio networks 
that are not laid out on a two dimensional grid? Are there other 
applications for the Bad Santa data streaming problem both in and 
outside the domain of radio networks?

\section{Acknowledgements}
We gratefully thank Kui Wu, Chiara Petrioli and James Horey for 
their helpful comments on issues of power consumption in radio 
networks.


\begin{thebibliography}{10} 

\comment{
\bibitem{along:frequency}
Noga Alon, Yossi Matias and Mario Szegedy. ``The Space Complexity 
of Approximating the Frequency Moments''. In the {\it Proceedings 
of the 28th Annual ACM Symposium on Theory of Computing (STOC)}, 
1996.

\bibitem{armstrong:wakeup}
Trevor Armstrong. ``Wake-up Based Power Management in Multi-hop 
Wireless Networks''. Unpublished manuscript available at 
www.eecg.toronto.edu/~trevor/Wakeup/index.html.

\bibitem{bhandari}
Vartika Bhandari and Nitin H. Vaidya. ``On Reliable Broadcast in a 
Radio Network''. Technical Report, CSL, UIUC, Feb. 2005.

\bibitem{bhandari2}
Vartika Bhandari and Nitin H. Vaidya. ``On Reliable Broadcast in a 
Radio Network: A Simplified Characterization''. Technical Report, 
CSL, UIUC, May 2005.

\bibitem{koo2}
Vartika Bhandhari, Jonathan Katz, Chiu-Yuen Koo and Nitin Vaidya. 
``Reliable Broadcast in Radio Networks: The Bounded Collision 
Case''. In {\it Proceedings of the $25^{th}$ Annual ACM Symposium 
on Principles of Distributed Computing (PODC)}. ACM Press, 2006.

\bibitem{braynard:asynchronous}
Rebecca Braynard and Carla Ellis. ``Asynchronous and Asymmetric 
Communication for Balancing Energy Consumption in Sensor 
Networks''. Technical Report CS-2004-11, Duke University, December 
2004.

\bibitem{cormode:space} Graham Cormode, Flip Korn, S. 
Muthukrishnan and Divesh Srivastava. Space- and Time-Efficient 
Deterministic Algorithms for Biased Quantiles Over Data Streams. 
In {\it Proceedings of the $25^{th}$ ACM SIGMOD-SIGACT-SIGART 
symposium on Principles of Database Systems}, pp. 263ö272, 2006.

\bibitem{demaine:frequency}
Eric Demaine, Alejandro L\'{o}pez-Ortiz, and Ian Munro. 
``Frequency Estimation of Internet Packet Streams with Limited 
Space''. In {\it Proceedings of the 10th Annual European Symposium 
on Algorithms (ESA 2002)}, pp. 348-360, 2002.

\bibitem{feeney:investigating} 
Laura Marie Feeney and Martin Nilsson. ``Investigating the Energy 
Consumption of a Wireless Network Interface in an Ad Hoc 
Networking Environment''. In {\it Proceedings of IEEE INFOCOM}, 
2001.

\bibitem{havinga:energy}
Paul Havinga and Gerard Smit. ``Energy-efficient TDMA Medium 
Access Control Protocol Scheduling''. In {\it Proceedings of the 
Asian International Mobile Computing Conference (AMOC)}, 2000.

\bibitem{henzinger}
Monika Henzinger, Prabhaker Raghavan and Sridar Rajagopalan. 
``Computing on Data Streams''. Technical Report, SRC-TN-1998-011 
May, 1998.

\bibitem{hill:system}
Jason Hill, Robert Szewczyk, Alec Woo, Seth Hollar, David E. 
Culler and Kristofer S. J. Pister. ``System Architecture 
Directions for Networked Sensors". In {\it Proceedings of the 
$9^{th}$ International Conference on Architectural Support for 
Programming Languages and Operating Systems (ASPLOS}, pp. 93-104, 
2000.

\bibitem{muthu}
S. Muthukrishnan. ` `Data Streams: Algorithms and Applications''. 
{\it Foundations and Trends in Theoretical Computer Science}, 
Volume 1, No. 2, 2005.

\bibitem{karlof}
Chris Karlof and David Wagner. ``Secure Routing in Wireless Sensor 
Networks: Attacks and Countermeasures''. {\it Ad Hoc Networks}, 
Volume 1, No. 2-3, pp. 293-315, 2003.

\bibitem{koo}
Chiu-Yuen Koo. ``Broadcast in Radio Networks Tolerating Byzantine 
Adversarial Behavior''. In {\it Proceedings of the $23^{rd}$ 
Annual ACM Symposium on Principles of Distributed Computing 
(PODC)}. ACM Press, pp. 275-282, 2004.

\bibitem{munro}
Ian Munro and Mike Paterson. ``Selection and Sorting with Limited 
Storage''. Theoretical Computer Science, 12, pp. 315-323, 1980.

\bibitem{pease}
Marshall C. Pease, Robert E. Shostak and Leslie Lamport. 
``Reaching Agreement in the Presence of Faults''. {\it Journal of 
the ACM}, 27(2), pp. 228-234, 1980.

\bibitem{stemm:measuring}
Mark Stemm and Randy H. Katz. ``Measuring and Reducing Energy 
Consumption of Network Interfaces in Hand-Held Devices''. {\it 
IEICE Transactions on Communications}, volume E80-B, pp. 
1125-1131, 1997.

\bibitem{suha:approximate}
Sudipto Guha and Andrew McGregor. ``Approximate Quantiles and the 
Order of the Stream''. In {\it ACM Symposium on Principles of 
Database Systems}, pp. 273-279, 2006.

\bibitem{suha:lower}
Sudipto Guha and Andrew McGregor. ``Lower Bounds for Quantile 
Estimation in Random-Order and Multi-Pass Streaming''. In {\it 
International Colloquium on Automata, Languages and Programming}, 
2007.

\bibitem{wang}
Lan Wang and Yang Xiao. ``A Survey of Energy-Efficient Scheduling 
Mechanisms in Sensor Networks''. {\it Mobile Networks and 
Applications} 11, pp. 723-740, 2006.

\bibitem{wang:realistic}
Qin Wang, Mark Hempstead and Woodward Yang. ``A Realistic Power 
Consumption Model for Wireless Sensor Network Devices''. In {\it 
Proceedings of the Third Annual IEEE Communications Society 
Conference on Sensor, Mesh and Ad Hoc Communications and Networks 
(SECON)} , 2006.

\bibitem{yao:probabilistic}
Andrew Yao. "Probabilistic Computations: Toward a Unified Measure 
of Complexity". In {\it Proceedings of Eighteenth IEEE Symposium 
on Foundations of Computer Science}, pp. 222-227, 1977. 

}


\bibitem{along:frequency}
Noga Alon, Yossi Matias and Mario Szegedy. ``The Space Complexity 
of Approximating the Frequency Moments''. In the {\it Proceedings 
of the 28th Annual ACM Symposium on Theory of Computing (STOC)}, 
1996.

\bibitem{armstrong:wakeup}
Trevor Armstrong. ``Wake-up Based Power Management in Multi-hop 
Wireless Networks''. Unpublished manuscript available at 
www.eecg.toronto.edu/~trevor/Wakeup/index.html.

\bibitem{bhandari}
Vartika Bhandari and Nitin H. Vaidya. ``On Reliable Broadcast in a 
Radio Network''. In {\it Proceedings of the $24^{th}$ Annual ACM 
Symposium on Principles of Distributed Computing (PODC)}. ACM 
Press, 2005.

\bibitem{bhandari2}
Vartika Bhandari and Nitin H. Vaidya. ``On Reliable Broadcast in a 
Radio Network: A Simplified Characterization''. Technical Report, 
CSL, UIUC, May 2005.

\bibitem{koo2}
Vartika Bhandhari, Jonathan Katz, Chiu-Yuen Koo and Nitin Vaidya. 
``Reliable Broadcast in Radio Networks: The Bounded Collision 
Case''. In {\it Proceedings of the $25^{th}$ Annual ACM Symposium 
on Principles of Distributed Computing (PODC)}. ACM Press, 2006.

\bibitem{braynard:asynchronous}
Rebecca Braynard and Carla Ellis. ``Asynchronous and Asymmetric 
Communication for Balancing Energy Consumption in Sensor 
Networks''. Technical Report CS-2004-11, Duke University, December 
2004.

\bibitem{cormode:space} Graham Cormode, Flip Korn, S. 
Muthukrishnan and Divesh Srivastava. Space- and Time-Efficient 
Deterministic Algorithms for Biased Quantiles Over Data Streams. 
In {\it Proceedings of the $25^{th}$ ACM SIGMOD-SIGACT-SIGART 
symposium on Principles of Database Systems}, pp. 263ö272, 2006.

\bibitem{demaine:frequency}
Eric Demaine, Alejandro L\'{o}pez-Ortiz, and Ian Munro. 
``Frequency Estimation of Internet Packet Streams with Limited 
Space''. In {\it Proceedings of the 10th Annual European Symposium 
on Algorithms (ESA 2002)}, pp. 348ö360, 2002.

\bibitem{feeney:investigating} 
Laura Marie Feeney and Martin Nilsson. ``Investigating the Energy 
Consumption of a Wireless Network Interface in an Ad Hoc 
Networking Environment''. In {\it Proceedings of IEEE INFOCOM}, 
2001.

\bibitem{havinga:energy}
Paul Havinga and Gerard Smit. ``Energy-efficient TDMA Medium 
Access Control Protocol Scheduling''. In {\it Proceedings of the 
Asian International Mobile Computing Conference (AMOC 2000)}, 
2000.

\bibitem{henzinger}
Monika Henzinger, Prabhaker Raghavan and Sridar Rajagopalan. 
``Computing on Data Streams''. Technical Report, SRC-TN-1998-011 
May 1998.

\bibitem{hill:system}
Jason Hill, Robert Szewczyk, Alec Woo, Seth Hollar, David E. 
Culler and Kristofer S. J. Pister. ``System Architecture 
Directions for Networked Sensors". In {\it Proceedings of the 
$9^{th}$ International Conference on Architectural Support for 
Programming Languages and Operating Systems (ASPLOS}, pp. 93-104, 
2000.

\bibitem{muthu}
S. Muthukrishnan. ``Data Streams: Algorithms and Applications''. 
{\it Foundations and Trends in Theoretical Computer Science}, 
Volume 1, No. 2, 2005.

\bibitem{karlof}
Chris Karlof and David Wagner. ``Secure Routing in Wireless Sensor 
Networks: Attacks and Countermeasures''. {\it Ad Hoc Networks}, 
Volume 1, No. 2-3, pp. 293-315, 2003.

\bibitem{koo}
Chiu-Yuen Koo. ``Broadcast in Radio Networks Tolerating Byzantine 
Adversarial Behavior''. In {\it Proceedings of the $23^{rd}$ 
Annual ACM Symposium on Principles of Distributed Computing 
(PODC)}. ACM Press, pp. 275-282, 2004.

\bibitem{munro}
Ian Munro and Mike Paterson. ``Selection and Sorting with Limited 
Storage''. Theoretical Computer Science, 12, pp. 315-323, 1980.

\bibitem{pease}
Marshall C. Pease, Robert E. Shostak and Leslie Lamport. 
``Reaching Agreement in the Presence of Faults''. {\it Journal of 
the ACM}, 27(2), pp. 228-234, 1980.

\bibitem{stemm:measuring}
Mark Stemm and Randy H. Katz. ``Measuring and Reducing Energy 
Consumption of Network Interfaces in Hand-Held Devices''. {\it 
IEICE Transactions on Communications}, volume E80-B, pp. 
1125-1131, 1997.

\bibitem{suha:approximate}
Sudipto Guha and Andrew McGregor. ``Approximate Quantiles and the 
Order of the Stream''. In {\it ACM Symposium on Principles of 
Database Systems}, pp. 273ö279, 2006.

\bibitem{suha:lower}
Sudipto Guha and Andrew McGregor. ``Lower Bounds for Quantile 
Estimation in Random-Order and Multi-Pass Streaming''. In {\it 
International Colloquium on Automata, Languages and Programming}, 
2007.

\bibitem{wang}
Lan Wang and Yang Xiao. ``A Survey of Energy-Efficient Scheduling 
Mechanisms in Sensor Networks''. {\it Mobile Networks and 
Applications} 11, pp. 723-740, 2006.

\bibitem{wang:realistic}
Qin Wang, Mark Hempstead and Woodward Yang. ``A Realistic Power 
Consumption Model for Wireless Sensor Network Devices''. In {\it 
Proceedings of the Third Annual IEEE Communications Society 
Conference on Sensor, Mesh and Ad Hoc Communications and Networks 
(SECON)} , 2006.

\bibitem{yao:probabilistic}
Andrew Yao. "Probabilistic Computations: Toward a Unified Measure 
of Complexity". In {\it Proceedings of Eighteenth IEEE Symposium 
on Foundations of Computer Science}, pp. 222-227, 1977.

\end{thebibliography}
\end{document}